\documentclass[conference]{IEEEtran}
\usepackage[boxruled]{algorithm2e}
\usepackage{cite}
\usepackage{textcomp}
\usepackage{graphicx}
\usepackage{psfrag}
\usepackage{subfigure}
\usepackage{url}
\usepackage{amsmath}
\usepackage{array}
\usepackage{amssymb}
\usepackage{amsfonts}
\usepackage{verbatim}
\usepackage{wasysym}

\newtheorem{theorem}{Theorem}

\title{Secrecy Capacity of the Gaussian Wire-Tap Channel with Finite Complex Constellation Input}

\begin{document}

\author{
\authorblockN{G. D. Raghava}
\authorblockA{Dept. of ECE, Indian Institute of Science \\
Bangalore 560012, India\\
Email:raghava@ece.iisc.ernet.in
}
\and
\authorblockN{B. Sundar Rajan}
\authorblockA{Dept. of ECE, Indian Institute of Science, \\Bangalore 560012, India\\
Email: bsrajan@ece.iisc.ernet.in
}
}

\maketitle
\thispagestyle{empty}	
%%%%%%%%

\begin{abstract}
The secrecy capacity of a discrete memoryless Gaussian Wire-Tap Channel when the input is from a finite complex constellation is studied. It is shown that the secrecy capacity curves of a finite constellation plotted against the SNR, for a fixed noise variance of the eavesdropper's channel has a \textit{global maximum} at an internal point. This is in contrast to what is known in the case of Gaussian codebook input where the secrecy capacity curve is a bounded, monotonically increasing function of SNR. Secrecy capacity curves for some well known constellations like BPSK, 4-QAM, 16-QAM and 8-PSK are plotted and the SNR at which the maximum occurs is found through simulation. It is conjectured that the secrecy capacity curves for finite constellations have a single maximum.  
\end{abstract}	

\section{Introduction}
Security is one of the most important issues in communication, and more so in wireless communications where the inherent broadcast nature of the wireless medium makes the data more susceptible to eavesdropping. The conventional techniques to secure data from wire-tappers are based on cryptographic techniques, where legitimate parties share a secret key. Shannon, in \cite{Sha}, analyzed such a system and showed that to achieve perfect secrecy the conditional probability of the \textit{encrypted data given a message} must be independent of the actual transmitted message. In \cite{Wyn}, Wyner applied this concept to the discrete memoryless channel with a wire-tapper where the wire-tapper's signal is a degraded version of the legitimate receiver's signal. His measure of secrecy is the \textit{normalized equivocation rate}, denoted by $\delta$, as seen by the wire-tapper. The rate-equivocation pairs ($R,\delta$) ($R$ is the rate of information transmission to the legitimate receiver in bits$/$channel use) were determined and the existence of positive \textit{secrecy rates, $C_s$}, for communication below which it is possible to limit the rate of information leaked to the wire-tapper to arbitrarily small values, was shown. In \cite{LH}, authors specialized Wyner's results to Gaussian Wire-Tap Channel (WTC) and showed that the secrecy capacity is the difference in the Gaussian capacities of the main channel and the eavesdropper's channel. Generalizations to the WTC were studied by Csis\'{a}r and K\"{o}rner in \cite{CK}, in which the Broadcast Channel was assumed to be more general (not necessarily degraded) and the sender wishes to transmit common information to both the legitimate receiver and the wire-tapper, in addition to the secret information to the legitimate receiver. Several other related studies have been reported on WTC$/$Broadcast Channel, Relay channel with confidential messages, Interference channel with confidential message, MIMO Gaussian wiretap channel etc \cite{LH}-\cite{LMYS}.

Coding for WTCs has been studied by several authors. In \cite{Wei}, Wei showed how to encode secret information using cosets of certain linear block codes. In \cite{TDC} Thangaraj et al. and in \cite{LLPS} Liu et al., extended this idea and showed that  low-density parity-check codes (LDPC) can asymptotically achieve the secrecy capacity of a erasure WTC and also showed how this coding method can be extended to achieve secrecy rates below the secrecy capacity for other channels. In \cite{KHM}, Demijan Klinc constructed LDPC codes for the Gaussian WTC. By transmitting the messages over the punctured bits, to hide it from the eavesdropper, they showed that whenever the eavesdropper's $SNR$ is below the $SNR$ threshold of the main channel (the gap which they call the  \textit{'security gap'}) BER at the eavesdropper increases to 0.5 within few iterations (BER at the eavesdropper being their measure of secrecy). Recently, construction of codes for the Gaussian WTC based on Lattices have also been reported \cite{BO}.

The secrecy capacity of a \textit{Gaussian WTC} was derived in  \cite{LH} and was shown that the Gaussian codebook achieves capacity. Though the secrecy capacity versus $SNR$ (at the legitimate receiver) plot gives us an idea of the achievable rates under various noise conditions of the eavesdropper's channel it fails to tell us what happens when finite-input alphabet restrictions are applied. The effect of finite complex constellation on the capacity has already been studied for Gaussian Multiple Access Channels\cite{Har}, Gaussian Broadcast channels\cite{Nav}, and Interference Channels\cite{FA}. In this paper, we study the effect of finite-input alphabet on the secrecy capacity of a Gaussian WTC.

The setting in this paper is a Gaussian WTC with the sender using finite complex constellation and the elements are chosen uniformly. Capacity calculations under the assumption that the input elements are chosen uniformly from a finite complex constellation is referred to as the \textit{Constellation Constrained (CC) Capacity} \cite{Big}. Throughout the paper, we refer to the secrecy capacity of a Gaussian WTC, calculated under the above assumptions, as \textit{Constellation Constrained Secrecy Capacity (CC-SC)}. Unquantized version of the received signal is assumed throughout the paper.

The main contributions of this paper are as follows:
\begin{itemize}
\item 
Analytically we show that the \textit{CC-SC} curves plotted against the $SNR$ of the main channel for a fixed noise variance of the eavesdropper's channel has a global maximum unlike that with the Gaussian codebook input.\\
\item 
We conjecture from the plots that the secrecy capacity curves for finite constellations have a single maximum. The $SNR$ at which the maximum occurs and the rate at this $SNR$ are found through simulations for some well known constellations like BPSK, 4-QAM, 16-QAM and 8-PSK.\\
\end{itemize}
This analysis, we hope, can be a guideline for practical code designers intending to code for the Gaussian WTC, to decide on the SNR at which to operate.\\

\textit{\textbf{Notations:}} For a random variable $X$ which takes values from the set $\cal S$, we assume some ordering of its elements and use $s^i$ to represent the $i$-th element of $S$. i.e. $s^i$ represents a value of the realization of the random variable $X$. Absolute value of a complex number $x$ is denoted by $|x|$ and $E[X]$ denotes the expectation of the random variable $X$. $H(X)$ denotes the entropy of a discrete random variable $X$ and $h(X)$ denotes the differential entropy of a continuous random variable $X$. For a complex random variable $N$, $N \sim {\cal CN}(0,{\sigma^2})$ denotes that $N$ has a circularly symmetric complex normal distribution with mean 0 and variance $\sigma^2$.\\

\section{System Model and Secrecy Capacity Calculations from Mutual Information}
\label{sec1}

The model of a Gaussian WTC is as shown in Fig \ref{fig:AWGNWT}. The sender wishes to convey information to the legitimate receiver in the presence of an eavesdropper. The goal of the sender here is, not just to convey information reliably to the legitimate receiver but also to perfectly secure the data from the eavesdropper. Sender is equipped with a finite complex constellation $\cal S'$ of size $M$ with power constraint $P_o$. When the sender transmits symbol $\tilde{X}$ from $\cal S'$, the received symbols $\tilde{Y}$ and $\tilde{Z}$ at the legitimate receiver and the eavesdropper respectively, are given by
\begin {align}
\label{eqnset1a}
\tilde{Y} = \tilde{X} + \tilde{N_1}\\
\label{eqnset1b}
\tilde{Z} = \tilde{X} + \tilde{N_2}
\end {align}
where $\tilde{X} \in \cal S'$, $\tilde{N_1} \sim {\cal CN}(0,{\sigma_1^2})$, $\tilde{N_2} \sim {\cal CN}(0,{\sigma_2^2})$.

By scaling the random variable $\tilde{X}$ by $\left(1/ \sqrt{P_o}\right)$, $\tilde{Y}$ and $\tilde{Z}$ by $\left(1/\sigma_1\right)$, the set of equations (\ref{eqnset1a}) and (\ref{eqnset1b}) can be equivalently represented as,   
\begin {align*}
Y = \sqrt{SNR}~X + N_1\\
Z = \sqrt{SNR}~X + N_2
\end {align*}
where $SNR$ = $\left(P_o/\sigma_1^2\right)$, $Y$ = $\left(1/\sigma_1\right)\tilde{Y}$, $Z$ = $\left(1/\sigma_1\right)\tilde{Z}$, $X\in\cal S$ (= normalized $\cal S'$), $N_1 \sim {\cal CN}(0,1)$, $N_2 \sim {\cal CN}(0,\sigma^2 \equiv \frac{\sigma_2^2}{\sigma_1^2})$.\\

\begin{figure}[htbp]
\centering
\includegraphics[totalheight=1in,width=2in]{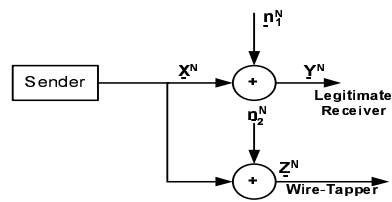}
\caption{AWGN-WT Model}	
\label{fig:AWGNWT}	
\end{figure}

The secrecy capacity $\mathcal{C}_s$ for a general wire tap channel\cite{CK} is given by 
\begin {equation*}
\mathcal{C}_s = \max_{V\rightarrow X\rightarrow (Y,Z)} [I(V ; Y ) - I(V ;Z)] 
\end {equation*}
where the maximum is over all possible random variables V in joint distribution with $X$, $Y$ and $Z$ such that $V\rightarrow X\rightarrow (Y,Z)$ is a Markov chain. The random variable V does not have a direct physical meaning; it is used
only for the calculation purpose. Calculation of secrecy capacity for a general discrete memoryless WTC is difficult and is still unsolved. However, by imposing certain restrictions on the characteristics of the main channel and the eavesdropper's channel, secrecy capacity can be calculated. In \cite{CK} it is shown that whenever the main channel is \textit{"less noisy"} compared to the eavesdropper's channel \cite{CK}, secrecy capacity is always positive and the calculation of which simplifies to
\begin {equation*}
\mathcal{C}_s = \max_{P_X(x)} [I(X ; Y ) - I(X ;Z)] 
\end {equation*}
where the maximum is over all possible distributions $P_X(x)$ of X.
We specialize our setting to the case of \textit{"less noisy main channel"} by imposing the restriction of  $\sigma_2^2 > \sigma_1^2$. Also, for the finite constellation setting since the distribution $P_X(x)$ of X is assumed to be uniform \cite{Big}, secrecy capacity calculation further simplifies to just the difference in the mutual informations of the main and the eavesdropper's channel. This secrecy capacity, which we refer to as \textit{CC-SC} is calculated as
\begin {equation}
\label{eqnset3g}
\mathcal{C}_s^{cc} = [I(X ; Y ) - I(X ;Z)]. 
\end {equation}
The mutual informations $I(X ; Y )$ and $I(X ;Z)$ are calculated as follows
\begin {equation}
\label{eqnset3h}
I(X ; Y ) = h(Y) - h(Y\vert X), \\
I(X ; Z ) = h(Z) - h(Z\vert X).
\end {equation}
Since X is assumed to be uniformly distributed and $N_1$ is Gaussian distributed,
$p(y \vert x = x^i), ~ p(y)$ and $h(Y\vert X)$ are given by,
%%%%%%%%%%%%%%%%
\begin {align}
\label{eqnset3i_a}
&p(y \vert x = x^i) = \frac{1}{\pi} \exp{\left(-|y - \left(\sqrt{SNR}\right)x^i|^2\right)},\\
\label{eqnset3i_b}
&~~~h(Y\vert X) ~=~ \frac {1}{M} \sum_{i=0}^{i=M-1}h(Y\vert X=x^i) \nonumber\\ 
&~~~~~~~~~~~~~~= ~h(N_1) ~= ~\log_2 {\pi e}, \\
\label{eqnset3i_c}
&~~~~~~~~~p(y) = \frac {1}{M}\sum_{i=0}^{i=M-1}p(y \vert x = x^i). 
\end {align}
%%%%%%%%%%%%%%%%%%%%%%%
Similarly, the expressions for $p(z\vert x=x^i), ~ p(z)$ and $h(Z\vert X)$  are obtained as,
%%%%%%%%%%%%
\begin {align}
\label{eqnset3j_a}
&p(z \vert x = x^i) = \frac{1}{\pi \sigma^2} \exp{\left(- \frac {|z - \left(\sqrt{SNR}\right)x^i|^2}{\sigma^2}\right)}, \\ 
\label{eqnset3j_b}
&~~~~h(Z\vert X) = \frac {1}{M} \sum_{i=0}^{i=M-1}h(Z\vert X=x^i)\notag\\
&~~~~~~~~~~~~~= ~h(N_2) ~= ~\log_2 {\pi e\sigma^2},  \\
\label{eqnset3j_c}
&~~~~~~~~p(z) = \frac {1}{M}\sum_{i=0}^{i=M-1}p(z\vert x = x^i). 
\end {align}
%%%%%%%%%%%%%%%
Using (\ref{eqnset3i_a}), (\ref{eqnset3i_b}), (\ref{eqnset3i_c}), and (\ref{eqnset3j_a}), (\ref{eqnset3j_b}), (\ref{eqnset3j_c}), expressions for $I(X;Y)$ and $I(X;Z)$ are found and are given in (12) and (13), at the top of the next page. Taking the difference of these two mutual informations we get an expression for the \textit{CC-SC} $\mathcal{C}_s^{cc}$ and is given in (14).\\  

\begin{figure*}
\centering
\begin {align*}
\label{eqnset3k}
& h(Y) ~~~=-\int_y\!p(y)\log_2{p(y)}\ dy \\
& ~~~~~~~~~~=~\log_2 M + \log_2{\pi}- \frac{1}{M}\sum_{i=0}^{i=M-1}E_{N_1}\left[\log_2{\left(\sum_{j=0}^{j=M-1}\exp{-\left(|N_1+\sqrt{SNR}\left(x^i-x^j\right)|^2\right)}\right)}\right]\\
&I(X;Y) = h(Y) - \log_2 {\pi e}\\
& ~~~~~~~~~~=~\log_2{\left(\frac{M}{e}\right)}-\frac{1}{M}\sum_{i=0}^{i=M-1}E_{N_1}\left[\log_2{\left(\sum_{j=0}^{j=M-1}\exp{-\left(|N_1+\sqrt{SNR}\left(x^i-x^j\right)|^2\right)}\right)}\right]  \tag{12} \\
& h(Z) ~~~= -\int_z\!p(z)\log_2{p(z)}\ dz \\
&~~~~~~~~~~=~\log_2 M + \log_2{\pi \sigma^2}- \frac{1}{M}\sum_{i=0}^{i=M-1}E_{N_2}\left[\log_2{\left(\sum_{j=0}^{j=M-1}\exp{-\left(\frac{|N_2+\sqrt{SNR}\left(x^i-x^j\right)|^2}{\sigma^2}\right)}\right)}\right]\\
&I(X;Z) = h(Z) - \log_2 {\pi e\sigma_2^2}\\
& ~~~~~~~~~~= ~\log_2{\left(\frac{M}{e}\right)}- \frac{1}{M}\sum_{i=0}^{i=M-1}E_{N_2}\left[\log_2{\left(\sum_{j=0}^{j=M-1}\exp{-\left(\frac{|N_2+\sqrt{SNR}\left(x^i-x^j\right)|^2}{\sigma^2}\right)}\right)}\right]  \tag{13} \\
&\mathcal{C}_s^{cc}~~~~~~= I(X;Y) - I(X;Z)\\
&~~~~~~~~~~=\frac{1}{M}\sum_{i=0}^{i=M-1}E_{N_2}\left[\log_2{\left(\sum_{j=0}^{j=M-1}\exp{-\left(\frac{|N_2+\sqrt{SNR}\left(x^i-x^j\right)|^2}{\sigma^2}\right)}\right)}\right]\\
&~~~~~~~~~~~~~~~~~~~~~~~~~~~~~~~~~~~~~~~~~~~~~~~~~~~-E_{N_1}\left[\log_2{\left(\sum_{k=0}^{k=M-1}\exp{-\left(|N_1+\sqrt{SNR}\left(x^i-x^k\right)|^2\right)}\right)}\right]\tag{14} \\
\end{align*}
\hrule
\end{figure*}

\begin{theorem}
\label{thm1}
The plot of \textit{CC-SC} $\mathcal{C}_s^{cc}$ versus $SNR$ of the main channel for a fixed noise variance of the eavesdropper's channel has a \textit{Global maximum} at an $SNR < \infty$.
\end{theorem}
\begin{proof}

\begin{comment}
&~~~~~~~~~~=\frac{1}{M}\left[\sum_{i=0}^{i=M-1}E_{N_2}\left[\log_2{\left(\sum_{j=0}^{j=M-1}\exp{-\left(\frac{|N_2+\left(\sqrt{SNR}\right)\left(x^i-x^j\right)|^2}{\sigma_2^2}\right)}\right)}\right]-E_{N_1}\left[\log_2{\left(\sum_{k=0}^{k=M-1}\exp{-\left(|N_1+\left(\sqrt{SNR}\right)\left(x^i-x^j\right)|^2}\right)}\right]\right]\tag{20}
\begin {align}
\label{eqnset3l}
\mathcal{C}_s^{cc} = [I(X ; Y ) - I(X ;Z)] 
\end {align}
\end{comment}
It is easy to see from (14) that for $SNR = 0, ~~ \mathcal{C}_s^{cc} = 0$. Also note that both $I(X;Y)$ and  $I(X ; Z)$ tends to $\log_2 M$ as  $SNR$ tends to $\infty$. Therefore,\\
\begin {align*}
\label{eqnset3m}
& \lim_{SNR\rightarrow \infty} \mathcal{C}_s^{cc} = \lim_{SNR\rightarrow \infty}[I(X ; Y ) - I(X ;Z)] \\
&~~~~~~~~~~~~~~= \log_2 M - \log_2 M = 0.
\end {align*}
Now since  $\mathcal{C}_s^{cc}$ is a bounded positive continuous function of $SNR$ we know from the maximum principle\cite{TT} that there exists a point $SNR_{max}$ less than $\infty$ on the $SNR$ axis at which $\mathcal{C}_s^{cc}$ attains a global maximum.
\end{proof}

Thus, from Theorem \ref{thm1}, \textit{CC-SC} function has a global maximum at a point less than $\infty$. This is in contrast to what happens when a Gaussian code book is used by the sender\cite{LH}. Let $\mathcal{C}_s^{GC}$ represent the secrecy capacity for the Gaussian WTC which, as shown in \cite{LH}, is achieved by using Gaussian codebook at the sender.  
\begin{align*}
&~~~~\mathcal{C}_s^{GC} = \log_2 {\left(1+\frac{P_o}{\sigma_1^2} \right)} - \log_2 {\left(1+\frac{P_o}{\sigma_2^2}\right)} \\ 
&~~~~~~~~~=\log_2{\frac{P_o+\sigma_1^2}{P_o + \sigma_2^2}} + \log_2{\frac{\sigma_2^2}{\sigma_1^2}}&~~~~~~~~~~\\ 
&As~P_o\rightarrow \infty, ~ we ~ have,\\
&~~~~~~~~~~~~\mathcal{C}_s^{GC} \rightarrow \log_2{\frac{\sigma_2^2}{\sigma_1^2}}
\end{align*}
Since $\log_2(.)$ is a monotonically increasing function of its argument, $\mathcal{C}_s^{GC}$ increases monotonically as $SNR$ increases and converges to
 $\log_2{\sigma_2^2 / \sigma_1^2}$. Thus the behavior of the secret capacity when finite constellation is used, is different from what happens when the Gaussian codebook is used. This is a very crucial point, which we feel is important and is missing in the $\mathcal{C}_s^{GC}$ plot, to consider when constructing practical codes for the Gaussian WTC. 

\section{ CC-SC Plots for Some Known Constellations }
\label{sec2}

\begin{figure}[htbp]
\centering
\includegraphics[totalheight=3in,width=3.15in]{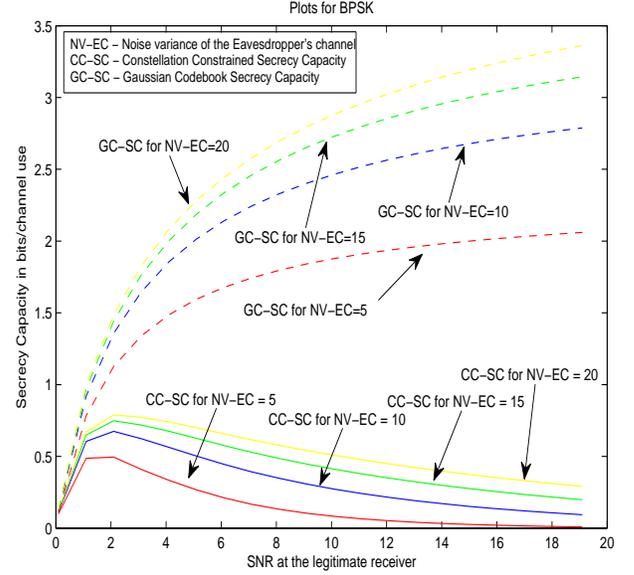}
\caption{CC Capacity curves when BPSK is used and for varying noise variances of the eavesdropper's channel. Curves are plotted for noise variances of 5, 10, 15 and 20}	
\label{fig:BPSK}	
\end{figure}

\begin{figure}[htbp]
\centering
\includegraphics[totalheight=3in,width=3.15in]{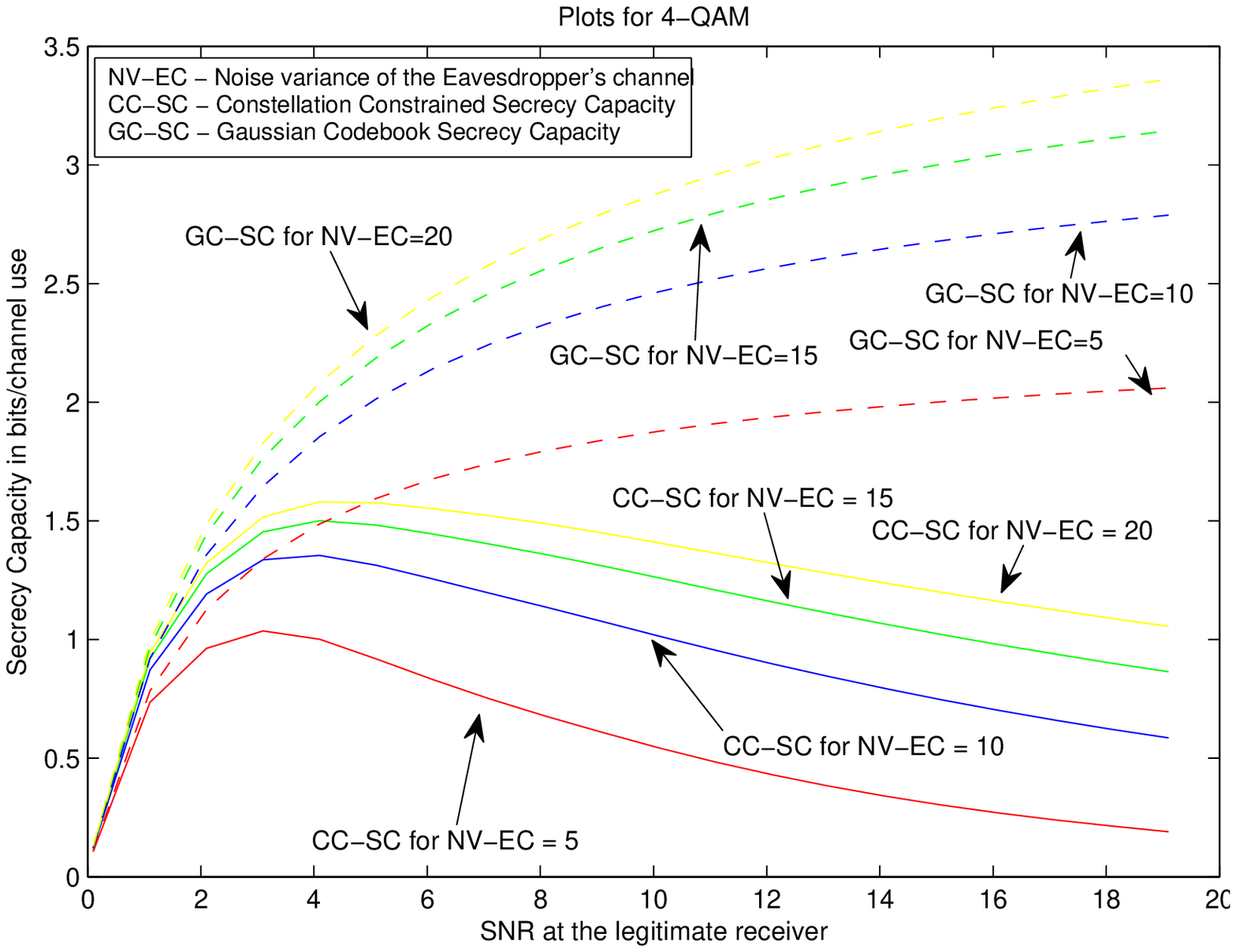}
\caption{CC Capacity curves when 4-QAM is used and for varying noise variances of the eavesdropper's channel. Curves are plotted for noise variances of 5, 10, 15 and 20.}	
\label{fig:4-QAM}	
\end{figure}

\begin{figure}[htbp]
\centering
\includegraphics[totalheight=3.1in,width=3.15in]{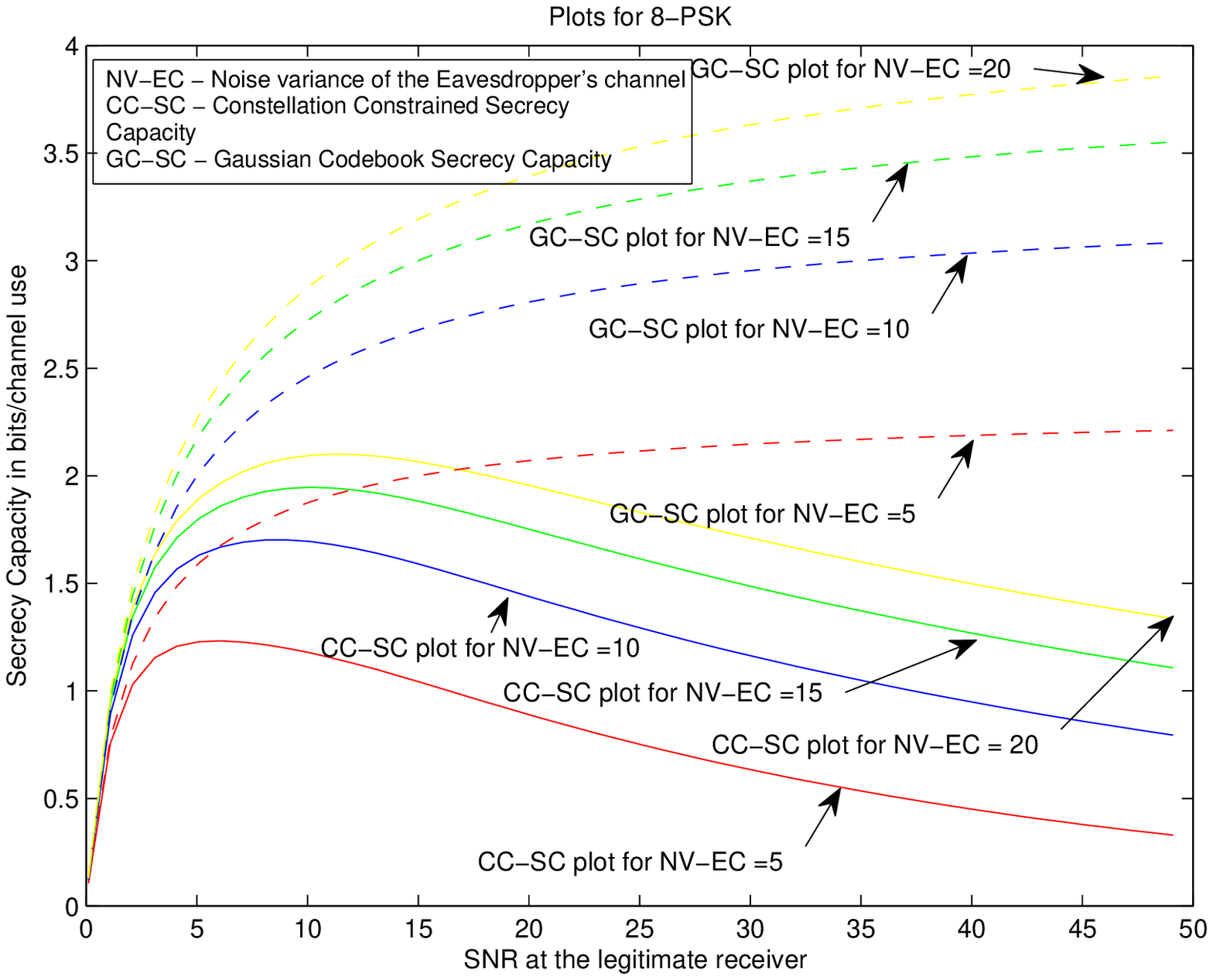}
\caption{CC Capacity curves when 8-PSK is used and for varying noise variances of the eavesdropper's channel. Curves are plotted for noise variances of 5, 10, 15 and 20.}	
\label{fig:8-PSK}	
\end{figure}

\begin{figure}[htbp]
\centering
\includegraphics[totalheight=3.1in,width=3.15in]{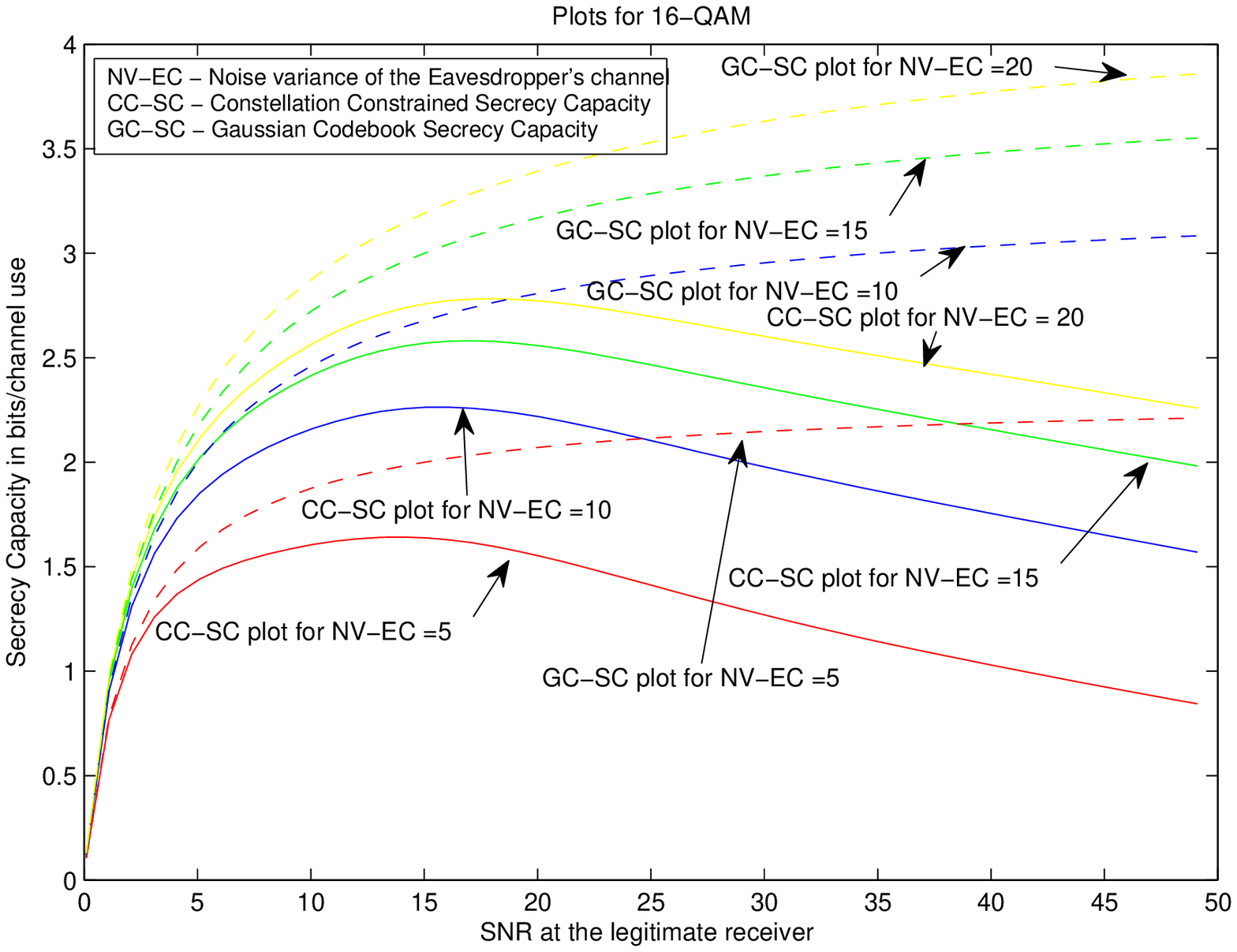}
\caption{CC Capacity curves when 16-QAM is used and for varying noise variances of the eavesdropper's channel. Curves are plotted for noise variances of 5, 10, 15 and 20.}	
\label{fig:16-qam}	
\end{figure}

\begin{figure}[htbp]
\centering
\includegraphics[totalheight=3in,width=3in]{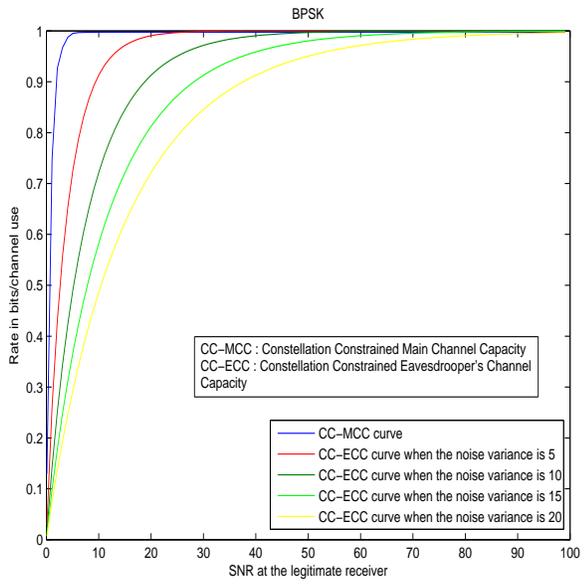}
\caption{\textit{CC} capacity plot of main channel and that of eavesdropper's channel for different $\sigma_2^2$ when BPSK is used.} 
\label{fig:bpsk_mi}	
\end{figure}

\begin{figure}[htbp]
\centering
\includegraphics[totalheight=3in,width=3in]{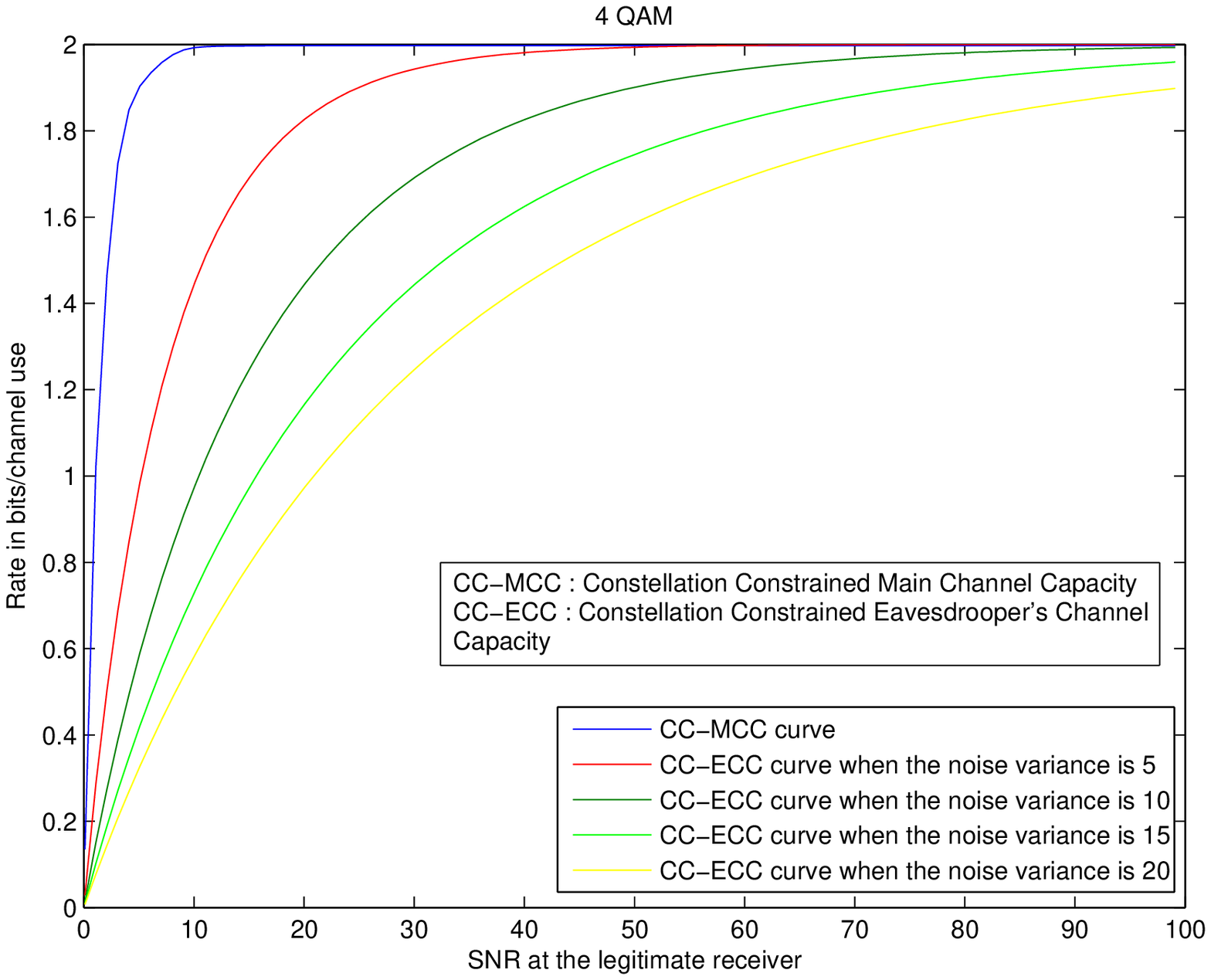}
\caption{\textit{CC} capacity plot of main channel and that of eavesdropper's channel for different $\sigma_2^2$ when 4-QAM constellation is used.}	
\label{fig:4-qam-mi}	
\end{figure}

\begin{figure}[htbp]
\centering
\includegraphics[totalheight=3in,width=3in]{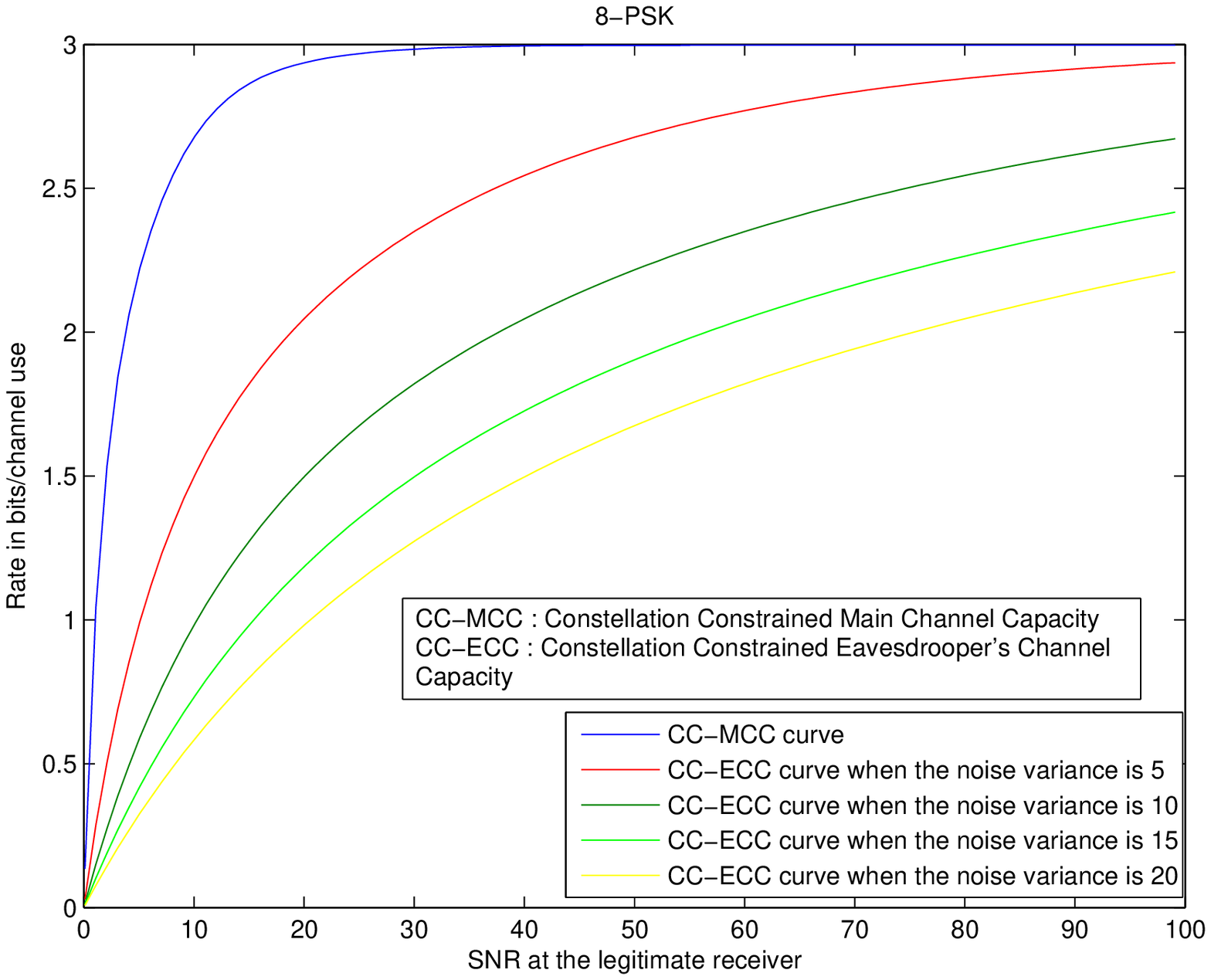}
\caption{\textit{CC} capacity plot of main channel and that of eavesdropper's channel for different $\sigma_2^2$ when 8-PSK constellation is used.}	
\label{fig:8-psk-mi}	
\end{figure}

\begin{figure}[htbp]
\centering
\includegraphics[totalheight=3in,width=3in]{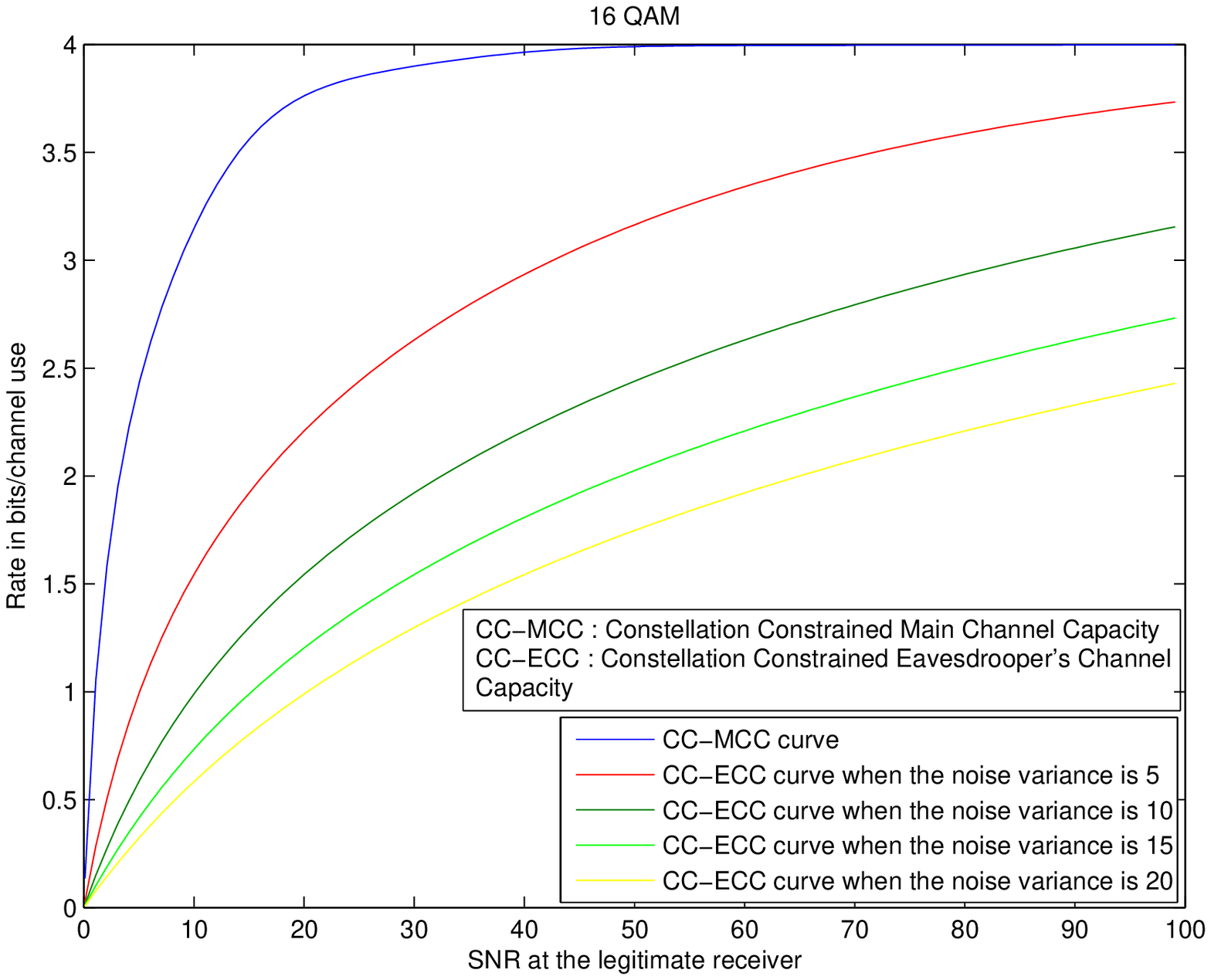}
\caption{\textit{CC} capacity plot of main channel and that of eavesdropper's channel for different $\sigma_2^2$ when 16-QAM constellation is used.}	
\label{fig:16-qam-mi}	
\end{figure}

\begin{figure}[htbp]
\centering
\includegraphics[totalheight=3in,width=3in]{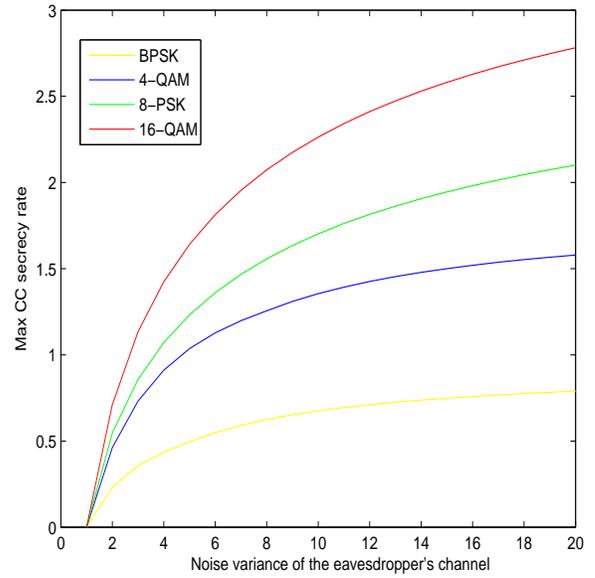}
\caption{Max CC Secrecy rate curves showing the maximum secrecy rate achieved as a function of the noise variance of the eavesdropper's channel.}	
\label{fig:maxrate}	
\end{figure}

\begin{figure}[htbp]
\centering
\includegraphics[totalheight=3in,width=3in]{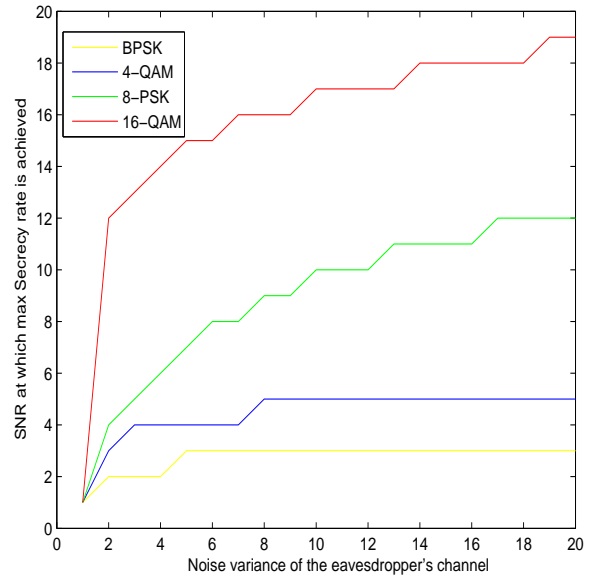}
\caption{$SNR$ curves showing the $SNR$ at which the maximum in the CC secrecy capacity is achieved as a function of the noise variance of the eavesdropper's channel.}	
\label{fig:SNR_at_max}	
\end{figure}

\begin{figure*}[htbp]
\centering
\includegraphics[totalheight=5in,width=5in]{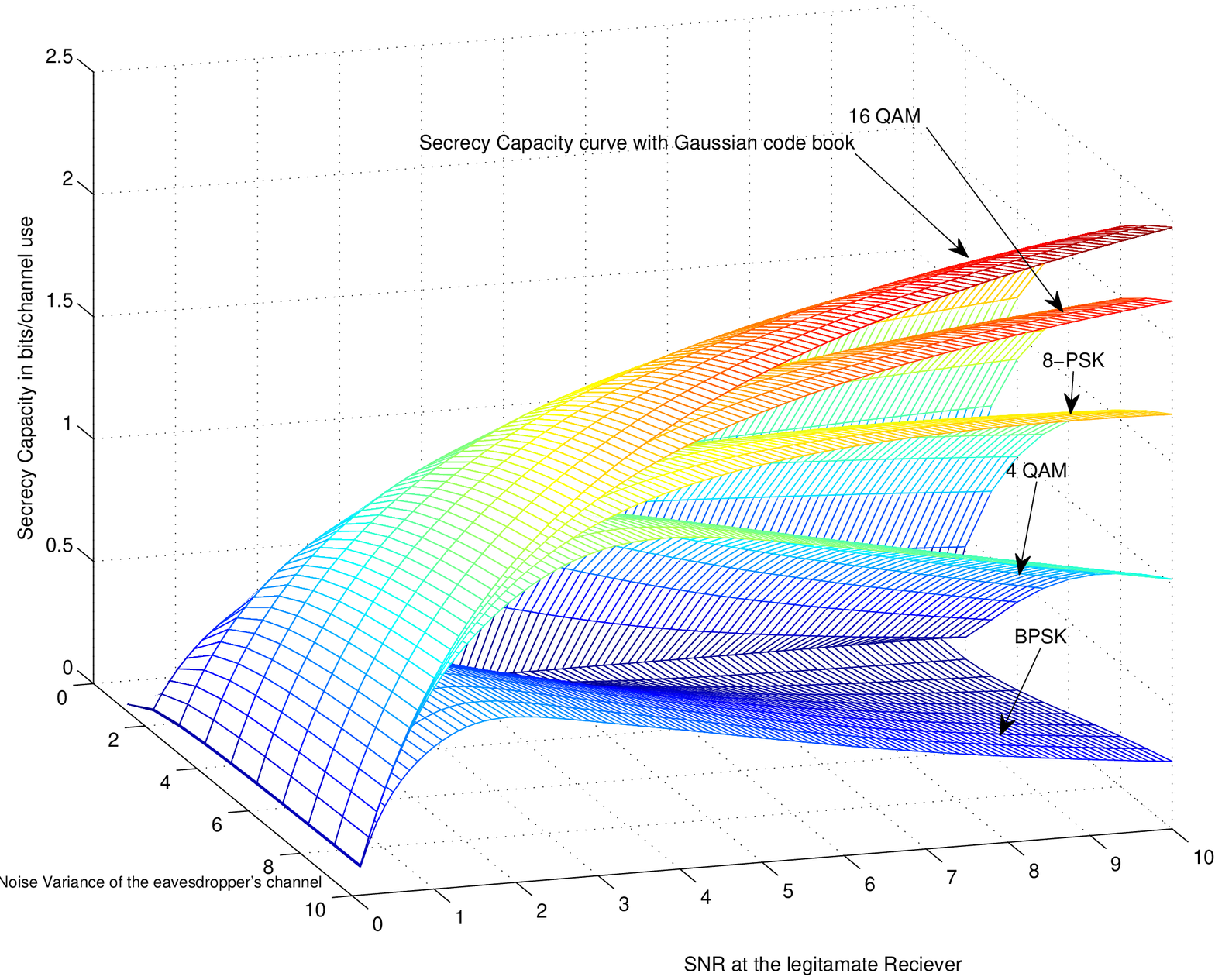}
\caption{Combined 3D-plot showing the \textit{CC-SC} curves for BPSK, 4-QAM, 8-PSK, 16-QAM.}	
\label{fig:Combined_Capacity_plots}	
\end{figure*}

All the figures in this section are plotted by fixing the noise variance of the main channel to 1 and increasing the average power of the constellation used continuously, to vary the $SNR$. Since the assumption is that the eavesdropper's channel must be more noisy compared to the main channel, $\sigma_2^2$ is $>1$ in all the figures. In each of Fig.\ref{fig:BPSK}, Fig.\ref{fig:4-QAM}, Fig.\ref{fig:8-PSK} and Fig.\ref{fig:16-qam} we fix the constellation and plot the curves by varying $\sigma_2^2$. Observe that the \textit{CC-SC} curves for different constellations shows up a maximum for some finite $SNR$ and monotonically decreases thereafter, before converging to zero at infinity. Also observe the behavior of the secrecy capacity curves with Gaussian codebook which continuously increases as the $SNR$ increases and converges to $\log_2{\frac{\sigma_2^2}{\sigma_1^2}}$ at infinity.

Fig.\ref{fig:bpsk_mi}, Fig.\ref{fig:4-qam-mi}, Fig.\ref{fig:8-psk-mi} and Fig.\ref{fig:16-qam-mi} are the plots showing the individual \textit{CC} channel capacities of the main channel and the eavesdropper's channel for various $\sigma_2^2$. From (\ref{eqnset3g}) we know that the \textit{CC-SC} is the difference between individual \textit{CC} capacities of the main channel and the eavesdropper's channel. Observe that CC capacity curves of the main channel and eavesdropper's channel start at $0$ and converge at infinity. As $\sigma_2^2$ increases the \textit{CC} capacity curve of the eavesdropper's channel rises more and more slowly and eventually converge to $\log_2 M$. It can also be observed that the gap between the \textit{CC} capacities of the main channel and eavesdropper's channel, which is also the \textit{CC-SC}, increases slowly, maximizes at some point and decreases thereafter. Based on this observation, we conjecture that the \textit{CC-SC} curves for various constellations have a single maximum. Gaussian codebook secrecy capacity plots show that not much is gained by increasing the power, but in fact from the \textit{CC-SC} plots we show that, just by increasing the power, for any fixed constellation, $\mathcal{C}_s^{cc}$, which is the maximum achievable secrecy rate, takes a hit as $SNR$ increases beyond the point of maximum.

Fig.(\ref{fig:maxrate}) is a plot showing the maximum achievable secrecy rate (Rate corresponding to the $SNR$ at which the maximum occurs), for different constellations, by varying $\sigma_2^2$. As it can be seen from the plots, maximum achievable rate monotonically increases and converges to the SISO Gaussian capacity, as $\sigma_2^2$ tends to infinity, at infinity. This is because as $\sigma_2^2$ tends to $\infty$, the WTC looks more like a SISO Gaussian channel and the maximum rate achievable in the SISO case is infinity.

Fig.(\ref{fig:SNR_at_max}) is a plot showing the $SNR$ at which the maximum occurs for different constellations, by varying $\sigma_2^2$. As can be seen from the plots since the $SNR$ at which the maximum occurs remains constant for large deviation in $\sigma_2^2$, exactly knowing $\sigma_2^2$ to operate at maximum is mitigated. But it is also true that the $SNR$ at which the maximum occurs increases continuously and tends to infinity as the noise variance tends to infinity. This is due to the fact that, as the noise variance tends to infinity, the WTC looks more like a SISO Gaussian channel with maximum rate occurring at $SNR$ $=\infty$.

Fig.(\ref{fig:Combined_Capacity_plots}) is a combined 3-D \textit{CC-SC} plots for BPSK, 4-QAM, 8-PSK, 16-QAM, plotted by varying the $SNR$ at the legitimate receiver and by varying $\sigma_2^2$. This plot combines the effect of changing $SNR$ as well as changing $\sigma_2^2$ on the \textit{CC-SC}. As can be seen from the plots, though the Gaussian codebook secrecy capacity plot continuously increases for increasing $SNR$ and increasing $\sigma_2^2$, the \textit{CC-SC} for all the constellations shows that as $\sigma_2^2$ increases, $SNR$ must be increased at a controlled rate, to be operating at the maximum. Either fixing the $SNR$ to one particular value or an uncontrolled change with respect to $\sigma_2^2$ makes \textit{CC-SC} to decrease continuously        

\section{Conclusion}
We have shown that the \textit{CC-SC} curve for a Gaussian WTC plotted against the $SNR$ of the main channel, for a fixed noise variance of the eavesdropper's channel has a global maximum. We also conjectured from the plots that the curves may have a single maximum. This result shows that when designing practical codes for a Gaussian WTC, for any chosen constellation, increasing the power beyond the maximum point is harmful as the secrecy capacity curve dips continuously thereafter. This also shows that, when designing a coding scheme, using some finite complex constellation, constrained by power $P_o$, the $SNR$ corresponding to $P_o$ if is greater than the $SNR_{max}$, it makes sense to reduce the power and operate at $SNR_{max}$. Thus we feel the result is very important more so because the capacity plots with Gaussian codebook fails to provide this information. The plots, we feel, might also be used to compare the performances of some of the existing coding schemes like the one based on LDPC in \cite{KHM}. In the case of LDPC based code construction given in \cite{KHM}, information about the maximum possible reduction in the security gap can be obtained. Thus our result can be a guideline for practical code constructions.    
\section{Acknowledgment}
This work was partly supported by the DRDO-IISc program
on Advanced Research in Mathematical Engineering, through
a research grant and by the INAE Chair Professorship grant to B.S. Rajan.

\end{document}